\documentclass{amsart}
\usepackage[left=1in, right=1in, top=1in, bottom=1in]{geometry}
\usepackage{enumitem}
\usepackage{graphicx}
\usepackage{amsmath, amsthm, amssymb}
\usepackage{xspace}
\usepackage{comment}
\usepackage[hidelinks]{hyperref}
\usepackage[dvipsnames]{xcolor}
\usepackage{algorithm}
\usepackage{algpseudocode}
\usepackage{makecell}
\usepackage{array}

\usepackage{tikz}
\usetikzlibrary{shapes, calc, decorations.pathreplacing, calligraphy}
\usepackage{subcaption}
\tikzstyle{myStyle}=[shape = circle, minimum size = 2pt, inner sep =1.5pt, outer sep = 0pt, draw, fill=white]

\newcommand{\lsay}[1]{{\color{red} Lora says: #1}}

\newcommand{\hsay}[1]{{\color{JungleGreen} Heather says: #1}}

\newcommand{\garnersay}[1]{{\color{ForestGreen} Garner says: #1}}
\newcommand{\nsay}[1]{{\color{purple} Nathan says: #1}}
\newcommand{\rsay}[1]{{\color{cyan} Reem says: #1}}
\newcommand{\elizabethsay}[1]{{\color{green!50!blue} Elizabeth says: #1}}
\newcommand{\isay}[1]{{\color{orange} Inne says: #1}}
\newcommand{\gracesay}[1]{{\color{violet} Grace says: #1}}
\newcommand{\elenasay}[1]{{\color{Maroon} Elena says: #1}}
\newcommand{\asay}[1]{{\color{olive} Alex says: #1}}

\newcommand{\msay}[1]{{\color{blue} Michael says: #1}}

\newcommand{\algprobm}[1]{\textsc{#1}\xspace}

\theoremstyle{plain}
\newtheorem{theorem}{Theorem}[section]
\newtheorem{proposition}[theorem]{Proposition}
\newtheorem{corollary}[theorem]{Corollary}
\newtheorem{lemma}[theorem]{Lemma}
\newtheorem{claim}[theorem]{Claim}
\newtheorem{observation}[theorem]{Observation}

\newtheorem*{thm:MainCounting}{Theorem \ref{thm:MainCounting}}
\newtheorem*{thm:MainMedian}{Theorem \ref{thm:MainMedian}}

\theoremstyle{definition}
\newtheorem{definition}[theorem]{Definition}
\newtheorem{remark}[theorem]{Remark}

\newcommand{\Thm}[1]{Theorem~\ref{#1}\xspace}

\newcommand{\Rmk}[1]{Remark~\ref{#1}\xspace}

\newcommand{\pb}[1]{\left(#1\right)}
\newcommand{\st}[1]{\left\{#1\right\}}

\title{Complexity and Enumeration in Models of Genome Rearrangement}

\thanks{We wish to thank the American Mathematical Society for organizing the Mathematics Research Community workshop where this work began. This material is based upon work supported by the National Science Foundation under Grant Number DMS   1641020 on the MRC wiki. ML was partially supported by J. A. Grochow's NSF award CISE-2047756 and the University of Colorado Boulder, Department of Computer Science Summer Research Fellowship. ML wishes to thank J.A. Grochow for helpful discussions regarding counting complexity. A preliminary version of this work appeared in COCOON 2023 \cite{bailey2023complexity}.} 

\author[Bailey et. al.]{Lora Bailey}
\address[Bailey]{Department of Mathematics, Grand Valley State University, Allendale, MI}
\email{baileylo@gvsu.edu}
\author[]{Heather Smith Blake}
\address[Blake]{Department of Mathematics and Computer Science, Davidson College, Davidson, NC}
\email{hsblake@davidson.edu}
\author[]{Garner Cochran}
\address[Cochran]{Department of Mathematics and Computer Science, Berry College, Mount Berry, GA}
\email{gcochran@berry.edu}
\author[]{Nathan Fox}
\address[Fox]{Department of Quantitative Sciences, Canisius University, Buffalo, NY}
\email{fox42@canisius.edu}
\author[]{Michael Levet$^*$}
\address[Levet]{Department of Computer Science, College of Charleston, Charleston, SC}
\email{levetm@cofc.edu}
\author[]{Reem Mahmoud}
\address[Mahmoud]{Division of Science, NYU Abu Dhabi, Abu Dhabi, UAE  (\textup{The majority of the work was completed while at} Department of Computer Science, Virginia Commonwealth University, Richmond, VA)}
\email{rm7230@nyu.edu}
\author[]{Elizabeth Bailey Matson}
\address[Matson]{The Division of Mathematics and Computer Science, Alfred University, Alfred, NY}
\email{matson@alfred.edu}
\author[]{Inne Singgih}
\address[Singgih]{Department of Mathematical Sciences, University of Cincinnati, Cincinnati, OH}
\email{inne.singgih@uc.edu}
\author[]{Grace Stadnyk}
\address[Stadnyk]{Department of Mathematics, Furman University, Greenville, SC}
\email{grace.stadnyk@furman.edu}
\author[]{Xinyi Wang}
\address[Wang]{Department of Computational Mathematics, Science, and Engineering, Michigan State University, East Lansing, MI}
\email{wangx249@msu.edu}
\author[]{Alexander Wiedemann}
\address[Wiedemann]{Department of Mathematics and Computational Data Science, Hamline University, Saint Paul, MN  (\textup{The majority of the work was completed while at} Department of Mathematics, Randolph--Macon College, Ashland, VA)}
\email{awiedemann01@hamline.edu}

\begin{document}
\maketitle
\let\thefootnote\relax\footnotetext{$^*$Corresponding author}

\begin{abstract}
In this paper, we examine the computational complexity of enumeration in certain genome rearrangement models. We first show that the \algprobm{Pairwise Rearrangement} problem in the Single Cut-and-Join model (Bergeron, Medvedev, \& Stoye, \textit{J. Comput. Biol.} 2010) is $\#\textsf{P}$-complete under polynomial-time Turing reductions. Next, we show that in the Single Cut or Join model (Feijao \& Meidanis, \textit{IEEE ACM Trans. Comp. Biol. Bioinf.} 2011), the problem of enumerating all medians ($\algprobm{\#Median}$) is logspace-computable ($\textsf{FL}$), improving upon the previous polynomial-time ($\textsf{FP}$) bound of Mikl\'os \& Smith (RECOMB 2015).
\end{abstract}

\section{Introduction}\label{sec:Introduction}

With the natural occurrence of mutations in genomes and the wide range of effects this can incite, scientists seek to understand the evolutionary relationship between species. Several discrete mathematical models have been proposed (which we discuss later) to model these mutations based on biological observations. Genome rearrangement models consider situations in which large scale mutations alter the order of the genes within the genome. Sturtevant~\cite{Sturtevant1917, Sturtevant1931} observed the biological phenomenon of genome rearrangement in the study of strains of \textit{Drosophila} (fruit flies), only a few years after he produced the first genetic map~\cite{Sturtevant1913}. Palmer~\&~Herbon~\cite{PalmerHerbon} observed similar phenomenon in plants. McClintock~\cite{mcclintock1951chromosome} also found experimental evidence of genes rearranging themselves, or ``transposing'' themselves, within chromosomes.  Subsequent to his work on \textit{Drosophila}, Sturtevant together with Novitski~\cite{SturtevantNovitski} introduced one of the first genome rearrangement problems, seeking a minimum length sequence of operations (in particular, so-called \textit{reversals}~\cite{HannenhalliPevzner}) that would transform one genome into another.

In this paper, we consider genome rearrangement models where each genome consists of directed edges, representing genes. Each directed edge receives a unique label, and each vertex has degree 1 or 2 (where we take the sum of both the in-degree and out-degree). There are no isolated vertices. Notably, each component in the associated undirected graph is either a path or a cycle.  Biologically, each component in the graph represents a chromosome. Paths correspond to linear chromosomes, such as in eukaryotes, and cycles correspond to circular chromosomes, which play a role in tumor growth~\cite{RaphaelPevzner}; see Figure~\ref{adjacency-fig}.

A genome model specifies the number of connected components (chromosomes), the types of components (linear, circular, or a mix of the two), and the permissible operations. The models we will consider allow for removing (cutting) and creating (joining) instances where two edges (genes) are incident, with certain models allowing for multiple cuts or joins to occur as part of a single operation. The \textit{reversal}~model~\cite{SturtevantNovitski}, for example, takes as input a genome consisting precisely of a single linear chromosome. In a now classical paper, Hannenhalli~\&~Pevzner~\cite{HannenhalliPevzner} exhibited a polynomial-time algorithm for computing the distance between two genomes in the reversal model. Later, those same authors generalized the reversal model to allow for multiple chromosomes and additional operations~\cite{HannenhalliPevzner2}. There are also several models that permit genomes which consist of both linear and circular chromosomes, including, for instance, the \textit{Single Cut or Join}~(\textit{SCoJ})~\cite{FeijaoMeidanis}, \textit{Single Cut-and-Join}~(\textit{SCaJ})~\cite{BergeronMedvedevStoye}, and \textit{Double Cut-and-Join}~(\textit{DCJ})~\cite{YancopoulosAttieFriedberg} models (see Section~\ref{sec:GenomeRearrangement} for a precise formulation). When choosing an appropriate model, it is important to balance biological relevance with computational tractability. This motivates the study of the computational complexity for genome rearrangement problems.

There are several natural genome rearrangement problems. We have already mentioned the \mbox{\algprobm{Distance}} problem, which asks for the minimum number of operations needed to transform one genome into another. Other natural problems include \algprobm{Pairwise Rearrangement}~(see Definition~\ref{def:DistancePairwiseRearrangement}), \mbox{\algprobm{Median} (Definition~\ref{def:Medians})}, \mbox{\algprobm{Median Scenarios}} (Definition~\ref{def:Medians}), \algprobm{Tree Labeling}, and \algprobm{Tree Scenarios}.  The focus of this paper is on the \mbox{\algprobm{Pairwise Rearrangement}} and \mbox{\algprobm{Median}} problems. We refer to \cite{MiklosSmithSamplingCounting} for more on \algprobm{Tree Labeling} and \algprobm{Tree Scenarios}. We summarize the known complexity-theoretic results in Table~\ref{fig:Fig1}.

\begin{table}[h]
    \begin{tabular}{l|l|l|l|l|}
    & \text{Reversal} & \text{SCoJ} & \text{SCaJ} & \text{DCJ} \\ \hline
    \algprobm{Distance} & T : in $\textsf{FP}$ \cite{HannenhalliPevzner} & T : in $\textsf{FP}$ \cite{BergeronMedvedevStoye} & T : in $\textsf{FP}$ \cite{FeijaoMeidanis} & T : in $\textsf{FP}$ \cite{BergeronMixtackiStoye} \\ \hline
    
    \makecell[l]{\algprobm{Pairwise}\\ \algprobm{Rearrangement}} & \makecell[l]{C: \textsf{\#P}-complete \\ C: In \textsf{FPRAS}} & T: in $\textsf{FP}$ \cite{MiklosKissTannier} & \makecell[l]{T: $\textsf{\#P}$-complete$^{*}$ \\ U: in/not in \textsf{FPRAS}} & \makecell[l]{C: \textsf{\#P}-complete \\ T: in \textsf{FPRAS} \cite{MiklosTannierDCJFPRAS}} \\ \hline

    \algprobm{Median} & \makecell[l]{T: not in $\textsf{FP}^{\dagger}$ \\T: not in $\textsf{FPRAS}^{\dagger}$} & \makecell[l]{T: in $\textsf{FP}$ \cite{MiklosSmithSamplingCounting} \\ T: in \textsf{FL}$^{*}$}& \makecell[l]{U: $\textsf{FP}$/$\textsf{NP}$-hard \\ U: in/not in $\textsf{FPRAS}$} & \makecell[l]{T: not in $\textsf{FP}^{\ddagger}$ \\T: not in $\textsf{FPRAS}^{\ddagger}$} \\ \hline

    \makecell[l]{\algprobm{Median}\\ \algprobm{Scenarios}} & \makecell[l]{T: not in $\textsf{FP}^{\dagger}$ \\T: not in $\textsf{FPRAS}^{\dagger}$} & \makecell[l]{T: $\textsf{\#P}$-complete \cite{MiklosSmith2019} \\ U: in/not in \textsf{FPRAS}}& \makecell[l]{U: $\textsf{FP}$/$\textsf{\#P}$-complete \\ U: in/not in $\textsf{FPRAS}$} & \makecell[l]{T: not in $\textsf{FP}^{\ddagger}$ \\T: not in $\textsf{FPRAS}^{\ddagger}$} \\ \hline

    \makecell[l]{\algprobm{Tree}\\ \algprobm{Labeling}} & \makecell[l]{T: not in $\textsf{FP}^{\dagger}$ \\T: not in $\textsf{FPRAS}^{\dagger}$} & \makecell[l]{U: $\textsf{FP}$/$\textsf{\#P}$-complete  \\ U: in/not in \textsf{FPRAS}}& \makecell[l]{U: $\textsf{FP}$/$\textsf{\#P}$-complete \\ U: in/not in $\textsf{FPRAS}$} & \makecell[l]{T: not in $\textsf{FP}^{\ddagger}$ \\T: not in $\textsf{FPRAS}^{\ddagger}$} \\ \hline

    \makecell[l]{\algprobm{Tree}\\ \algprobm{Scenarios}} & \makecell[l]{T: not in $\textsf{FP}^{\dagger}$ \\T: not in $\textsf{FPRAS}^{\dagger}$} & \makecell[l]{T: $\textsf{\#P}$-complete \cite{MiklosSmith2019} \\ T: not in \textsf{FPRAS} \cite{MiklosKissTannier}}& \makecell[l]{U: $\textsf{FP}$/$\textsf{\#P}$-complete \\ U: in/not in $\textsf{FPRAS}$} & \makecell[l]{T: not in $\textsf{FP}^{\ddagger}$ \\T: not in $\textsf{FPRAS}^{\ddagger}$} \\ \hline
    \end{tabular}
    \caption{Theorems are denoted by T, conjectures by C, and problems with unknown complexity by U (those problems without enough evidence for a conjecture). The entry ``not in $\textsf{FP}$'' is under the assumption that $\textsf{P} \neq \textsf{NP}$. The entry ``not in $\textsf{FPRAS}$'' is under the assumption that $\textsf{RP} \neq \textsf{NP}$. Those marked with $\dagger$ and $\ddagger$ follow from the fact that the corresponding decision problem is $\textsf{NP}$-hard, \cite{Caprara1999} and  \cite{TannierZhengSankokff} respectively. Those marked with $^*$ indicate results in this paper.}
    \label{fig:Fig1}
\end{table}

\noindent \\ \textbf{Main Results.} Our first main result concerns the computational complexity of the \algprobm{Pairwise Rearrangement} problem in the Single Cut-and-Join model:

\begin{theorem} \label{thm:MainCounting} \label{thm:SCJCount} 
In the Single Cut-and-Join model, the
\algprobm{Pairwise Rearrangement} problem is \mbox{$\#\textsf{P}$-complete} under polynomial-time Turing reductions.
\end{theorem}

\begin{remark}
We establish \Thm{thm:MainCounting} in the special case when the adjacency graph (see Definition~\ref{def:adjgraph}) is a disjoint union of cycles. A related question that remains open is  whether this $\#\textsf{P}$-completeness holds when the adjacency graph consists of only paths. 
\end{remark}

We also improve the known computational complexity of the $\#\algprobm{Median}$ problem in the Single Cut or Join model. Mikl\'os \& Smith \cite{MiklosSmithSamplingCounting} previously showed that counting the number of medians---the $\#\algprobm{Median}$ problem---belongs to $\textsf{FP}$. We improve this complexity-theoretic upper bound as follows:

\begin{theorem} \label{thm:MainMedian} \label{thm:MedianSCJ}
In the Single Cut or Join model, the $\#\algprobm{Median}$ problem belongs to $\textsf{FL}$.
\end{theorem}

\begin{remark} \label{rmk:MedianComplexity}
While it is widely believed that $\textsf{FP} \neq \#\textsf{P}$, no such separation is known. In particular, Allender~\cite{Allender1999ThePR} showed that the permanent is not computable using uniform $\textsf{TC}^{0}$ circuits, which implies  $\textsf{TC}^{0} \neq \#\textsf{P}$. While there has been subsequent work on lower bounds against the permanent for both threshold and arithmetic circuits \cite{KoiranPerifel, JansenSanthanam2012, JansenSanthanam2013}, it remains open as to whether even $\textsf{NC}^{1} \neq \textsf{\#P}$. It is known that $\textsf{TC}^{0} \subseteq \textsf{NC}^{1} \subseteq \textsf{L}$. While it is believed that each of these containments are strict, no such separation is known. 
\end{remark}

\noindent \\ \textbf{Further Related Work.} There has been significant work on efficient computational approaches, such as sampling  and approximation (see, for instance,~\cite{MiklosKissTannier, MiklosSmithSamplingCounting, DarlingMiklosRagan, DurrettNielsenYork, MiklosTannier, LargetSimonKadaneSweet}), to cope with the intractability of enumeration. In addition to the problems in Table~\ref{fig:Fig1}, we run into issues of combinatorial explosion when examining statistics such as the breakpoint reuse~\cite{Alekseyev2010, BergeronMixtacki} and the size/positions of reversals~\cite{AjanaLefebvre, DarlingMiklosRagan}.  Developing an efficient uniform or near-uniform sampler would allow for obtaining a statistically significant sample for hypothesis testing. Such samples are needed, for instance, to test the Random~Breakpoint~Model~\cite{Alekseyev2010, BergeronMixtacki} and check if there is natural selection for maintaining balanced
replichors~\cite{DarlingMiklosRagan}. Jerrum,~Valiant,~\&~Vazirani~\cite{JERRUM1986169} showed that for self-reducible problems, finding a near-uniform sampler has the same complexity as approximately enumerating the size of the space. Thus, approximate counting and sampling are closely related. Past work on these samplers has often utilized a rapidly mixing Markov chain on the full evolutionary history space~\cite{MiklosSmith2019}. 

The complexity of the \algprobm{Pairwise Rearrangement} problem in the reversal model remains an intriguing open question. Sorting scenarios in the reversal model correspond to so-called \textit{pressing sequences} on an appropriate vertex-colored graph \cite{HannenhalliPevzner}. There has been considerable work in studying these pressing sequences \cite{BixbyFlintMiklos, CooperDavis, CooperWhitlatch, CooperGartlandWhitlatch, CooperHannaWhitlatch, WhitlatchThesis}.

While our work in this paper focuses on genomes without repeated genes, consideration for gene duplications has biological motivations. Indeed, gene duplications are widespread events and have been recognized as driving forces of evolution~\cite{Bailey2006PrimateSD, LynchOriginnGenome}. In human genomes, segmental duplications are common sites for non-allelic homologous recombination that lead to genomic disorders, copy-number polymorphisms, and gene and transcript innovations~\cite{JiangTangVenturaCardone}. It is $\textsf{NP}$-hard to compute the distance in the presence of duplicate genes for both the reversal~\cite{ReversalDuplicateGenes} and DCJ \cite{KececiogluSankoffDCJ} models. There has been work on heuristics using optimization techniques~\cite{ReversalDuplicateGenes, LaohakiatLursinapSuksawatchon, SuksawatchonLursinapBoden, DCJDuplicate}, resulting in computer packages that work well in practice~\cite{MSOAR, MSOAR2}. There has also been work in the Single Cut or Join model in the presence of duplicate genes---see, for instance,~\cite{SCoJDuplicate} and the references therein.

\section{Preliminaries}

\subsection{Genome Rearrangement} \label{sec:GenomeRearrangement}

For a standard reference, see \cite{FertinLabarreRusaTannierVialette}.

\begin{definition}
    A \emph{genome} is an edge-labeled directed graph in which each label is unique and the total degree of each vertex is 1 or 2 (in-degree and out-degree combined). In particular, a genome consists of disjoint paths and cycles. The components of a genome we call \emph{chromosomes}. Each edge begins at its \emph{tail} and ends at its \emph{head}, collectively referred to as its \emph{extremities}.  Degree $2$ vertices are called \emph{adjacencies}, and degree $1$ vertices are called \emph{telomeres}. {See Figure~\ref{adjacency-fig}.}
\end{definition}

\begin{figure}[!h]
\centering
\begin{tikzpicture}
\draw[ultra thick, >=stealth] (0,1) edge[<-] (1,0) (1,0) edge[->] (2,1) (2,1) edge[->] (3,0);
\draw[ultra thick, >=stealth] (4,0) edge[<-] (6,0) (6,0) edge [->] (5,1) (5,1) edge [<-] (4,0);
\draw (0.7,0.8) node {$X_1$} (1.8,0.2) node {$X_2$} (2.8,0.8) node {$X_3$} (4.2,0.8) node {$X_4$} (5.8,0.8) node {$X_5$} (5.1,-0.4) node {$X_6$};
\draw (-1,1) edge[->, >=stealth] (-0.1,1) (-2.2,1) node {telomere \footnotesize{$X_1^h$}} (6.1,0) edge[<-, >=stealth] (7,0) (8.5,0) node {adjacency \footnotesize{$X_5^tX_6^t$}}; 
\draw[thick, decorate, decoration={calligraphic brace, amplitude=3mm}] (-0.1,1.2) -- (3.1,1.2) (1.5,1.7) node {linear chromosome}; 
\draw[thick, decorate, decoration={calligraphic brace, amplitude=3mm}] (3.9,1.2) -- (6.1,1.2) (5,1.7) node {chromosome} (5,2) node {circular};
\draw[thick, decorate, decoration={calligraphic brace, mirror, amplitude=3mm}] (0,-0.5) -- (6,-0.5) (3,-1) node {Genome};
\draw[thick, decorate, decoration={calligraphic brace, mirror, amplitude=3mm}] (-0.2,0.8) -- (0.8,-0.2) (-0.1, 0) node[rotate=-45] {gene};
\end{tikzpicture}
\caption{An edge-labeled genome.}
\label{adjacency-fig}
\end{figure}
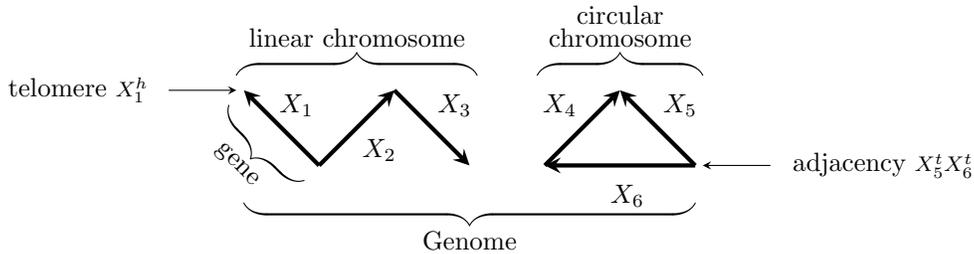

Adjacencies can be viewed as unordered sets of two extremities, and telomeres as sets containing exactly one extremity. For simplicity, we write adjacency $\{a,b\}$ as $ab$ and telomere $\{c\}$ as $c$.  For example, the adjacency $X_{5}^{t} X_{6}^{t}$ in Figure~\ref{adjacency-fig} denotes that the tail of the edge $X_{5}$ and the tail of the edge $X_{6}$ meet, and the telomere $X_1^h$ is where the edge $X_{1}$ ends. Each genome is then uniquely defined by its set of adjacencies and telomeres.

Consider the following operations on a given genome: 
\begin{enumerate}[itemsep=0pt]
    \item[ (i)] \emph{Cut}: an adjacency $ab$ is separated into two telomeres, $a$ and $b$,
    \item[ (ii)] \emph{Join}: two telomeres $a$ and $b$ become one adjacency, $ab$,
    \item[ (iii)] \emph{Cut-join}:  we replace adjacency $ab$ and telomere $c$ with adjacency $ac$ and telomere $b$, and
    \item[ (iv)] \emph{Double-cut-join}: we replace adjacencies $ab$ and $cd$ with adjacencies $ac$ and $bd$.
\end{enumerate}

\begin{figure}[h!]
\centering
\begin{subfigure}{0.7\textwidth}
\centering
\begin{tikzpicture}
\draw[ultra thick, >=stealth, xshift=-0.5cm] (0,1) edge[<-] (1,0) (1,0) edge[->] (2,1) (2,1) edge[->] (3,0);
\draw[xshift=-0.5cm] (0.2,0.2) node {$X_1$} (1.8,0.2) node {$X_2$} (2.8,0.8) node {$X_3$};
\draw (3.5,0.5) edge [->, >=stealth] (5.5,0.5) (4.5,0.2) node {cut} (4.5,0.8) node {(i)};
\draw[ultra thick, >=stealth, xshift=0.5cm] (6,1) edge[<-] (7,0) (7,0) edge[->] (8,1) (9,1) edge[->] (9,0);
\draw[xshift=0.5cm] (6.2,0.2) node {$X_1$} (7.8,0.2) node {$X_2$} (8.6,0.6) node {$X_3$}; 
\end{tikzpicture}
\end{subfigure}

\vspace{0.5cm}
\begin{subfigure}{0.7\textwidth}
\centering
\begin{tikzpicture}
\draw[ultra thick, >=stealth, xshift=-0.5cm] (0,1) edge[<-] (1,0) (1,0) edge[->] (2,1) (2,1) edge[->] (3,0);
\draw[xshift=-0.5cm] (0.2,0.2) node {$X_1$} (1.8,0.2) node {$X_2$} (2.8,0.8) node {$X_3$};
\draw (3.5,0.5) edge [->, >=stealth] (5.5,0.5) (4.5,0.2) node {join} (4.5,0.8) node {(ii)};
\draw[ultra thick, >=stealth, xshift=0.5cm] (6,1) edge[<-] (7,0) (7,0) edge[->] (8,1) (8,1) edge[->] (6,1);
\draw[xshift=0.5cm] (6.2,0.2) node {$X_1$} (7.8,0.2) node {$X_2$} (7,0.75) node {$X_3$};
\draw[white, xshift=0.5cm] (8,1) edge[->] (9,0);
\end{tikzpicture}
\end{subfigure}

\begin{subfigure}{0.7\textwidth}
\centering
\begin{tikzpicture}
\draw[ultra thick, >=stealth, xshift=-0.5cm] (0,1) edge[<-] (1,0) (1,0) edge[->] (2,1) (2,1) edge[->] (3,0);
\draw[xshift=-0.5cm] (0.2,0.2) node {$X_1$} (1.8,0.2) node {$X_2$} (2.8,0.8) node {$X_3$};
\draw (3.5,0.5) edge [->, >=stealth] (5.5,0.5) (4.5,0.2) node {cut-join} (4.5,0.8) node {(iii)};
\draw[ultra thick, >=stealth, xshift=0.5cm] (7,1) edge[<-, bend right=60, looseness=1.2] (7,0) (7,0) edge[->, bend right=60, looseness=1.2] (7,1) (8,1) edge[->] (9,0);
\draw[xshift=0.5cm] (6.2,0.2) node {$X_1$} (7.8,0.2) node {$X_2$} (8.8,0.8) node {$X_3$};
\draw[white, xshift=0.5cm] (7,1.5) node {$X_3$};
\end{tikzpicture}
\end{subfigure}

\begin{subfigure}{0.7\textwidth}
\centering
\begin{tikzpicture}
\draw[ultra thick, >=stealth, xshift=-0.5cm] (0,1) edge[<-] (1,0) (1,0) edge[->] (2,1) (2,1) edge[->] (3,0);
\draw[xshift=-0.5cm] (0.2,0.2) node {$X_1$} (1.8,0.2) node {$X_2$} (2.8,0.8) node {$X_3$};
\draw (3.5,0.5) edge [->, >=stealth] (5.5,0.5) (4.5,0.2) node {double-cut-join} (4.5,0.8) node {(iv)};
\draw[ultra thick, >=stealth, xshift=0.5cm] (6,1) edge[<-] (7,0) (8,1) edge[->] (7,0) (8,1) edge[->] (9,0);
\draw[xshift=0.5cm] (6.2,0.2) node {$X_1$} (7.8,0.2) node {$X_2$} (8.8,0.8) node {$X_3$}; 
\draw[white, xshift=0.5cm] (7,1.5) node {$X_3$};
\end{tikzpicture}
\end{subfigure}
\caption{(i) Adjacency $X_2^hX_3^t$ is cut. (ii) Telomeres $X_1^h$ and $X_3^h$ are joined. (iii) Adjacency $X_2^hX_3^t$ is cut, and resulting telomere $X_2^h$ is joined with $X_1^h$. (iv) Adjacencies $X_1^tX_2^t$ and $X_2^hX_3^t$ are replaced with $X_1^tX_2^h$ and $X_2^tX_3^t$.} 
\label{operations-fig}
\end{figure}
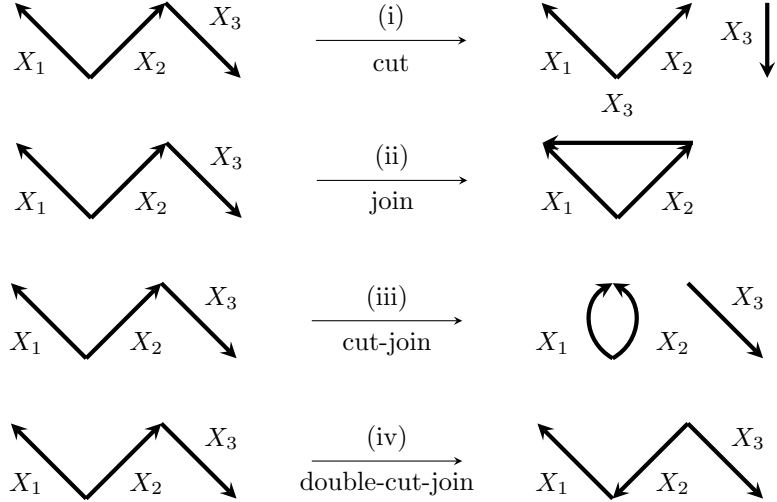

Note that a cut-join operation combines a single cut and a join into one operation, and a double-cut-join operation performs two cuts and two joins in one operation. See Figure~\ref{operations-fig} for an illustration of these operations.

Several key models are based on these operations.  The \textit{Double Cut-and-Join~(DCJ)}~model was initially introduced by Yancopoulos,~Attie,~\&~Friedberg~\cite{YancopoulosAttieFriedberg} and permits all four operations. Later, \mbox{Feijao~\&~Meidanis~\cite{FeijaoMeidanis}} introduced the \textit{Single~Cut~or~Join~(SCoJ)}~model, which only allows operations (i) and (ii). Alternatively, the \textit{Single~Cut-and-Join~(SCaJ)}~model~\cite{BergeronMedvedevStoye} allows operations (i)-(iii), but not operation (iv). 
In this paper, we consider the Single Cut-and-Join and Single Cut or Join models.

\begin{definition}
    For any genome rearrangement model $J$, it is always possible to perform a sequence of operations from $J$ that transforms genome $G_1$ into $G_2$ if they share the same set of edge labels.  Such a sequence is called a \emph{scenario}. 
   The minimum length of such a scenario is called the \emph{distance} and is denoted $d^{J}(G_1, G_2)$. When $J$ is understood, we simply write $d(G_{1}, G_{2})$. 
   An operation on a genome $G_1$ that (strictly) decreases the distance to genome $G_2$ is called a \emph{sorting operation} for $G_1$ and $G_2$. A scenario requiring $d(G_{1}, G_{2})$ operations to transform $G_{1}$ into $G_{2}$ is called a \emph{most parsimonious scenario} or \emph{sorting scenario}. When $G_{2}$ is understood, we refer to the action of transforming $G_{1}$ into $G_{2}$ using the minimum number of operations as \textit{sorting} $G_{1}$. The number of most parsimonious scenarios transforming $G_{1}$ into $G_{2}$ is denoted $\#\text{MPS}(G_{1}, G_{2})$. 
\end{definition}

We now turn to defining the key algorithmic problems that we will consider in this paper.

\begin{definition} \label{def:DistancePairwiseRearrangement}
Let $J$ be a model of genome rearrangement, and let $G_{1}$ and $G_{2}$ be genomes. The \algprobm{Distance} problem asks to compute $d(G_{1}, G_{2})$. The \algprobm{Pairwise Rearrangement} problem asks to compute $\#$MPS$(G_1,G_2)$. 
\end{definition}

\begin{definition} \label{def:Medians}
Let $J$ be a model of genome rearrangement, and let $\mathcal{G}$ be a collection of $k \geq 3$ genomes. A \emph{median} for $\mathcal{G}$ is a genome $G$ that minimizes
\[
\sum_{G_{i} \in \mathcal{G}} d(G_{i}, G).
\] 
\noindent The \algprobm{Median} problem asks for one median for $\mathcal{G}$.
 The \algprobm{\#Median} problem asks for the number of medians for $\mathcal{G}$. The \algprobm{Most Parsimonious Median Scenarios} problem, or \algprobm{Median Scenarios} problem, asks for the number of tuples $(G_{m}, \{ \sigma_{G} : G \in \mathcal{G}\})$, where $G_{m}$ is a median for $\mathcal{G}$ and $\sigma_{G}$ is a most parsimonious scenario transforming $G_{m}$ into $G$.
\end{definition}

To investigate these computational problems, we begin by introducing the adjacency graph. 

\begin{definition}\label{def:adjgraph}
    Given two genomes $G_1$ and $G_2$ with the same set of edge labels, the \emph{adjacency graph} $A(G_1,G_2)$ is a bipartite multigraph $(V_{1} \dot \cup V_{2}, E)$ 
    where each vertex in $V_i$ corresponds to a unique adjacency or telomere in $G_i$ and the number of edges between two vertices is the size of the intersection of the corresponding adjacency or telomere sets. 
\end{definition}

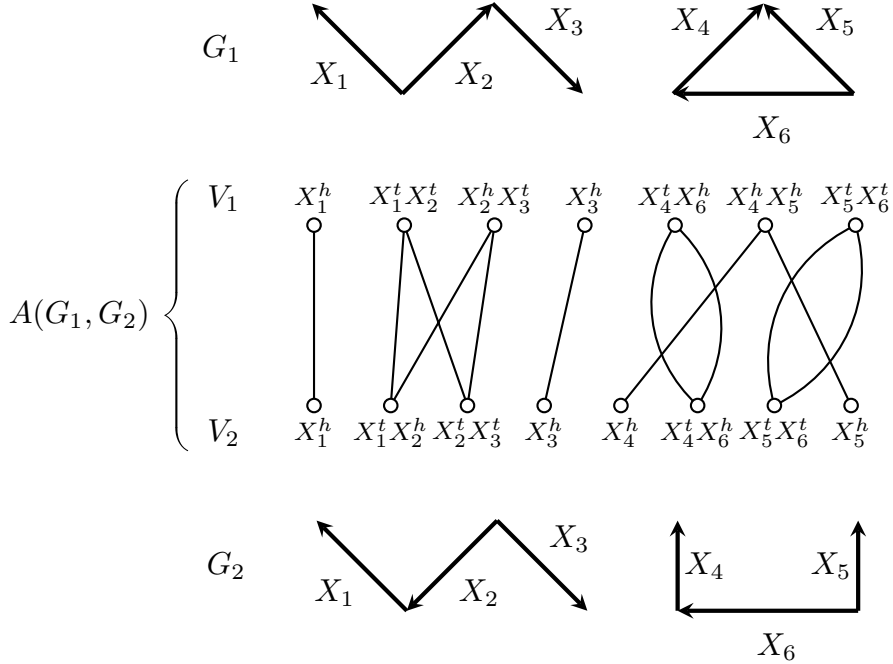
\begin{figure}[!h]
\centering
\begin{minipage}{0.8\textwidth}
\centering
\hspace{2.1cm}\begin{tikzpicture}[scale=1.2, every node/.style={scale=1.2}]
\draw[ultra thick, >=stealth] (0,1) edge[<-] (1,0) (1,0) edge[->] (2,1) (2,1) edge[<-] (3,0) (3,0) edge[<-] (4,1);
\draw[ultra thick, >=stealth] (5,0) edge[<-] (7,0) (7,0) edge [->] (6,1) (6,1) edge [<-] (5,0);
\draw (0.2,0.2) node {$X_1$} (1.8,0.2) node {$X_2$} (2.8,0.8) node {$X_3$} (3.8,0.2) node {$X_4$} (5.2,0.8) node {$X_5$} (6.8,0.8) node {$X_6$} (6.1,-0.4) node {$X_7$};
\draw[thick] (-1,0.5) node {$G_1$};
\end{tikzpicture}
\end{minipage}

\begin{minipage}{0.8\textwidth}
\centering
\begin{tikzpicture}[scale=1.2, every node/.style={scale=1.2}]
\draw[thick] (0,0) -- (0,-2) (3,0)-- (1.9,-2) -- (1,0) -- (0.9,-2) -- (2,0) --(2.9,-2) --(3,0)  (4,0) -- (3.75,-2) (5.25,-2) edge[bend right=40] (5,0) (5,0) edge[bend right=40] (5.25,-2) (4.4,-2) -- (6,0) -- (6.95,-2) (7,0) edge[bend left=40] (6.1,-2) (7,0) edge[bend right=40] (6.1,-2); 
\draw[thick] (0,0) node [myStyle] {} (1,0) node [myStyle] {} (2,0) node [myStyle] {} (3,0) node [myStyle] {} (4,0) node [myStyle] {} (5,0) node [myStyle] {} (6,0) node [myStyle] {} (7,0) node [myStyle] {};
\draw[thick] (0,-2) node [myStyle] {} (0.9,-2) node [myStyle] {} (1.9,-2) node [myStyle] {} (2.9,-2) node [myStyle] {} (3.75,-2) node [myStyle] {} (4.4,-2) node [myStyle] {} (5.25,-2) node [myStyle] {} (6.1,-2) node [myStyle] {} (6.95,-2) node [myStyle] {};
\draw[thick, white] (-1.5,0.5) node {$G_1$};
\draw[thick] (0,0.3) node {\footnotesize{$X_1^h$}} (1,0.3) node {\footnotesize{$X_1^tX_2^t$}} (2,0.3) node {\footnotesize{$X_2^hX_3^h$}} (3,0.3) node {\footnotesize{$X_3^tX_4^h$}} (4,0.3) node {\footnotesize{$X_4^t$}} (5,0.3) node {\footnotesize{$X_5^tX_7^h$}} (6,0.3) node {\footnotesize{$X_5^hX_6^h$}} (7,0.3) node {\footnotesize{$X_6^tX_7^t$}};
\draw[thick] (0,-2.3) node {\footnotesize{$X_1^h$}} (0.9,-2.3) node {\footnotesize{$X_1^tX_2^h$}} (1.9,-2.3) node {\footnotesize{$X_2^tX_3^t$}} (2.9,-2.3) node {\footnotesize{$X_3^hX_4^h$}} (3.75,-2.3) node {\footnotesize{$X_4^t$}} (4.4,-2.3) node {\footnotesize{$X_5^h$}} (5.25,-2.3) node {\footnotesize{$X_5^tX_7^h$}} (6.1,-2.3) node {\footnotesize{$X_6^tX_7^t$}} (6.95,-2.3) node {\footnotesize{$X_6^h$}};
\draw (-1.0,0.3) node {$V_1$};
\draw  (-1.0,-2.3) node {$V_2$};
\draw[thick, decorate, decoration={calligraphic brace, amplitude=3mm}] (-1.4,-2.5) -- (-1.4,0.5) (-2.6,-1.0) node {$A(G_1,G_2)$}; 
\end{tikzpicture}
\end{minipage}

\begin{minipage}{0.8\textwidth}
\centering
\hspace{2.1cm}\begin{tikzpicture}[scale=1.2, every node/.style={scale=1.2}]
\draw[ultra thick, >=stealth] (0,1) edge[<-] (1,0) (2,1) edge[->] (1,0) (2,1) edge[->] (3,0) (3,0) edge[<-] (4,1);
\draw[ultra thick, >=stealth] (5.3,0) edge[<-] (6.5,0) (6.5,0) edge [->] (7,1) (4.8,1) edge [<-] (5.3,0);
\draw (0.2,0.2) node {$X_1$} (1.8,0.2) node {$X_2$} (2.8,0.8) node {$X_3$} (3.8,0.2) node {$X_4$} (5.4,0.5) node {$X_5$} (6.4,0.5) node {$X_6$} (6.1,-0.4) node {$X_7$};
\draw[thick] (-1,0.5) node {$G_2$};
\draw[thick, white] (0.5,1.4) node {$X_6$};
\end{tikzpicture}
\end{minipage}
\caption{An adjacency graph $A(G_1,G_2)$ is shown in the middle, with genomes $G_1$ and $G_2$ shown above and below, respectively.}
\label{fig:adj graph}
\end{figure}

Note that each vertex in an adjacency graph $A(G_1,G_2)$ must have either degree 1 or 2 (corresponding, respectively, to telomeres and adjacencies in the original genome), and so  $A(G_1,G_2)$ is composed entirely of disjoint cycles and paths. Note also that every operation on $G_1$ corresponds to an operation on $V_1$ in $A(G_1,G_2)$. For example, in Figure~\ref{fig:adj graph} the cut operation on $G_1$ which separates adjacency $X_5^hX_6^h$ into telomeres $X_5^h$ and $X_6^h$ equates to separating the corresponding vertex $X_5^hX_6^h$ in $V_1$ into two vertices $X_5^h$ and $X_6^h$, thus splitting the path of length 2 into two disjoint paths of length 1 in $A(G_1,G_2)$. In a similar fashion, a join operation on $G_1$ corresponds to combining two vertices $a$, $b$ in $V_1$ into a single vertex $ab$, and a cut-join operation on $G_1$ corresponds to replacing vertices $ab$, $c$ in $V_1$ with vertices $ac$, $b$. Whether or not an operation on $A(G_{1}, G_{2})$ corresponds to a sorting operation on $G_1$---that is, whether it decreases the distance to $G_2$ or not---depends highly on the structure of the components acted on. To better describe such sorting operations, we adopt the following classification:

\begin{definition} Components of $A(G_{1}, G_{2})$ are classified as follows, where the \emph{size} of a component $B$ is defined to be $\lfloor \, |E(B)|/2\, \rfloor$:
\begin{itemize}
    \item A $W$-\emph{shaped component} is an even path with its two endpoints in $V_1$. 
    \item An $M$-\emph{shaped component} is an even path with its two endpoints in $V_2$.
    \item An $N$-\emph{shaped component} is an odd path, further called a \emph{trivial path} if it is size 0 (a single edge).
    \item A \emph{crown} is an even cycle, further called a \emph{trivial crown} if it is size 1 (a $2$-cycle).
\end{itemize}  \end{definition}

The language ``trivial'' is motivated by the fact that such components indicate where $G_1$ and $G_2$ already agree, and hence no sorting operations are required on vertices belonging to trivial components (see, e.g., the trivial components in Figure~\ref{fig:adj graph}). Indeed, a sorting scenario can be viewed as a minimal length sequence of operations which produces an adjacency graph consisting of only trivial components.

\begin{observation} \label{obs:SortingScenarios}
\noindent In the SCaJ model, a case analysis yields precisely these sorting operations on $A(G_1,G_2)$: 
\begin{enumerate}[itemsep=0pt]
\item[(a)] A cut-join operation on a non-trivial $N$-shaped component, producing an $N$-shaped component and a trivial crown
\item[(b)] A cut-join operation on a $W$-shaped component of size at least 2, producing a trivial crown and a $W$-shaped component
\item[(c)] A join operation on a $W$-shaped component of size 1, producing a trivial crown
\item[(d)] A cut operation on an $M$-shaped component, producing two $N$-shaped components
\item[(e)] A cut operation on a non-trivial crown, producing a $W$-shaped component
\item[(f)] A cut-join operation on an $M$-shaped component and a $W$-shaped component, where an adjacency in the $M$-shaped component is cut and joined to a telomere in the $W$-shaped component, producing two $N$-shaped components

 \item[(g)] A cut-join operation on a non-trivial crown and an $N$-shaped component, where an adjacency in the crown is cut and joined to the telomere in the $N$-shaped component, producing an $N$-shaped component
\item[(h)] A cut-join operation on a non-trivial crown and a $W$-shaped component, where an adjacency from the crown is cut and joined with a telomere from the $W$-shaped component, producing a $W$-shaped component
\end{enumerate}
\end{observation}

Note that (a) - (e) are sorting operations on $G_1$ that operate on only one component in the adjacency graph, though they may produce two different components. On the other hand, (f) - (h) are sorting operations on $G_1$ that operate on two separate components in the adjacency graph. 

Using these sorting operations on $A(G_1,G_2)$, the distance between two genomes $G_1$ and $G_2$ for the SCaJ model is given by
\begin{equation}
    d(G_1,G_2) = n - \frac{\#N}{2} - \#T + \#C
\end{equation}
where $n$ is the number of genes in $G_1$ (equivalently, one half of the number of edges in $A(G_1,G_2)$), $\#N$ is the number of $N$-shaped components, $\#T$ is the number of trivial crowns, and $\#C$ is the number of non-trivial crowns \cite{BergeronMedvedevStoye}.

Let $\mathcal{B}$ be the set of all components of $A(G_1,G_2)$ and let $\mathcal{B}'$ be a subset of $\mathcal{B}$. Define
\begin{align} 
d(\mathcal{B}'):=\left(\sum_{B\in\mathcal{B}'}\text{size}(B)\right)-\#T_{\mathcal{B}'}+\#C_{\mathcal{B}'} \label{eq:SubsetDistance}
\end{align} 
where $\#T_{\mathcal{B}'}$ and $\#C_{\mathcal{B}'}$ are the number of trivial crowns and nontrivial crowns in $\mathcal{B}'$, respectively. The quantity $d(\mathcal{B}')$ is the minimum number of operations needed to transform all components of $\mathcal{B}'$ into trivial components, with no operation acting on a component not belonging to $\mathcal{B}'$. Note that $d(\mathcal{B})=d(G_1,G_2)$, as the $\frac{\#N}2$ term is absorbed into the summation of the sizes of all components.

\begin{definition}\label{def:equivrelation}
Let $A$ and $B$ be components of an adjacency graph, and consider a particular sorting scenario. We say $A\sim B$ if either $A=B$ or there is a cut-join operation in the scenario where an extremity $a$ from $A$ and an extremity $b$ from $B$ are joined into an adjacency. The transitive closure of $\sim$ is an equivalence relation which we call \emph{sort together}. We will be particularly interested in subsets of the equivalence classes of ``sort together.'' We abuse terminology by referring to such a subset as a set that \emph{sorts together}.
\end{definition}

Note that if two components in $A(G_1, G_2)$ sort together, the cut-join witness of this does not need to occur immediately. For example, two non-trivial crowns $C_1$ and $C_2$ can sort together by first cutting $C_1$ to produce a $W$-shaped component, then applying operation (b) multiple times before using operation (h) to sort $C_2$ and the remaining $W$-shaped component together.

We will now introduce additional notation that we will use in this paper. Let $\mathcal{B}$ be the collection of all components of a given adjacency graph $A(G_{1}, G_{2})$. Let $\Pi(\mathcal{B})$ denote the set of all partitions of $\mathcal{B}$. Define $\#\text{MPS}(\mathcal{B})$ to be the number of most parsimonious scenarios transforming $G_{1}$ into $G_{2}$. For a partition $\pi \in \Pi(\mathcal{B})$, define $\#\text{MPS}(\mathcal{B}, \pi)$ to be the number of most parsimonious scenarios transforming $G_{1}$ into $G_{2}$, where two components $A$ and $B$ belong to the same part of $\pi$ if and only if $A$ and $B$ sort together. For a subset $\mathcal{B}'$ of $\mathcal{B}$, let $\#\text{ST}( \mathcal{B}')$ denote the number of sequences with $d(\mathcal{B}')$ operations in which the components of  $\mathcal{B}'$ sort together and are transformed into trivial components with no operation acting on a component not belonging to $\mathcal{B}'$.

We conclude by restricting our attention to a single component of an adjacency graph and determine the number of most parsimonious scenarios which sort that component, independent of the other components. We will later use these counts as building blocks to enumerate the most parsimonious scenarios for multiple components in the  adjacency graph.

\begin{lemma}
Let $A(G_{1}, G_{2})$ be an adjacency graph with component $B$. 
\begin{enumerate}[label=(\alph*)]
    \item If $B$ is an $N$-shaped component, then $\#\text{ST}(\{B\}) = 1$. 
    \item If $B$ is a $W$-shaped component of size $w$, then $\#\text{ST}(\{B\})=2^{w-1}$.
    \item If $B$ is a $M$-shaped component of size $m$, then $\#\text{ST}(\{B\}) = 2^{m-1}$.
    \item If $B$ is a non-trivial crown of size $c$, $\#\text{ST}(\{B\}) = c\cdot 2^{c-1}$.
\end{enumerate}
\end{lemma}

\begin{proof}
First, suppose $B$ is a non-trivial $N$-shaped component of size $n$.  There is only one sorting operation: cut-join with the telomere in $V_1$, which creates an $N$-shaped component of size $n-1$ and a trivial crown. The result follows as there is only one available sorting operation at each step.

Next, suppose $B$ is a non-trivial $W$-shaped component of size $w$. If $w=1$, there is only one sorting operation: join the two telomeres, yielding a trivial crown. 
 If $w>1$, there are two possible sorting operations: cut-join with either the left or right telomere, either of which produces a $W$-shaped component of size $w-1$ and a trivial crown. The result follows by induction on $w$.
 
 By relabeling $G_1$ and $G_2$, we see that the sorting scenarios for a $W$-shaped component of size $w$ are in bijection with the sorting scenarios for an $M$-shaped component of size $w$. Hence the result for $M$-shaped components matches the result for $W$-shaped components.

 Finally, suppose $B$ is a non-trivial crown of size $c$. The first sorting operation must be a cut of any of the $c$ adjacencies in $V_1$.  The resulting component is a $W$-shaped component of size $c$, for which there are $2^{c-1}$ scenarios. The result follows.
\end{proof}

\subsection{Computational Complexity}

We assume familiarity with standard notions such as $\textsf{P}$ and $\textsf{NP}$. For standard references, see~\cite{AroraBarak, ComplexityZoo}. A \emph{logspace transducer} is a $3$-tape deterministic Turing Machine with a read-only input tape, a work tape for which we can use $O(\log n)$ cells, and a write-only output tape where the tape-head never moves left. It is well-known, but not obvious, that the composition of two logspace transducers is a logspace transducer (see \cite{Sipser13} for a proof). A decision problem $A$ belongs to $\textsf{L}$ if there exists a logspace transducer $M$ that decides $A$. In particular, for decision problems $A, B$, if $B \in \textsf{L}$ and $A$ is logspace reducible to $B$, then $A \in \textsf{L}$. Denote by $\textsf{FP}$ the class of polynomial-time computable functions, and denote by $\textsf{FL}$ the class of logspace-computable functions. We begin by introducing key background necessary for our $\textsf{\#P}$-completeness result (Theorem~\ref{thm:MainCounting}).

\begin{definition}
We say that the function $f : \{0,1\}^{*} \to \mathbb{N}$ belongs to $\textsf{\#P}$ if there exists a non-deterministic Turing machine $M$ that runs in polynomial time such that $f(x)$ returns the number of accepting computations of $M$ on input $x$.
\end{definition}

\begin{remark}
Intuitively, for a language (decision problem) $L \in \textsf{NP}$, $\textsf{\#P}$ contains the functions that count the number of witnesses certifying that the given string $x$ belongs to $L$. There is a subtlety in that a function $f : \{0,1\}^{*} \to \mathbb{N}$ belonging to $\textsf{\#P}$ is defined based on the underlying Turing machine (algorithm) as opposed to the language itself.
\end{remark}

There are several notions of counting reductions. We recall them here:

\begin{definition}
Let $f, g \in \textsf{\#P}$. A \emph{parsimonious reduction} from $f$ to $g$ is a polynomial-time computable function $\varphi : \{0,1\}^{*} \to \{0,1\}^{*}$ such that for all $x \in \{0,1\}^{*}$, $f(x) = g(\varphi(x))$. We write $f \leq_{par}^{\textsf{P}} g$ to indicate that there is a parsimonious reduction from $f$ to $g$.
\end{definition}

\begin{definition}[{\cite{Zank1991PCompletenessVM}}]
Let $f, g \in \textsf{\#P}$. We say that there exists a \emph{weak parsimonious reduction} from $f$ to $g$, denoted $f \leq_{\text{wpar}}^{\textsf{P}} g$, if there exist polynomial-time computable functions $\varphi : \{0,1\}^{*} \to \{0,1\}^{*}$ and $\psi : \mathbb{N} \to \mathbb{N}$ such that $f(x) = (\psi \circ g \circ \varphi)(x)$ for all $x \in \{0,1\}^{*}$. 
\end{definition}

Intuitively, under a weak parsimonious reduction the number of witnesses for $x$ (which is $f(x)$) should be the same as some function of the number of witnesses for $\varphi(x)$. Noting that the number of witnesses for $\varphi(x)$ is $g(\varphi(x))$, we thus require $f(x) = \psi(g(\varphi(x)))$. Observe as well that a parsimonious reduction is a special case of a weak parsimonious reduction: take $\psi$ to be the identity function.

We next introduce the notion of a Turing reduction. To this end, we need the notion of an oracle. We refer the reader to \cite{AroraBarak} for the full technicalities of oracle Turing machines; for our purposes, the key intuitive notion suffices. Let $\mathcal{O}$ be a computable function, which we refer to in this context as an \emph{oracle}. An \emph{oracle Turing machine} is a Turing machine that is parameterized by an oracle, in which the Turing machine is able to query the image of a given string under the parameter oracle. When the oracle is unspecified, we write $M^{\square}$. When the oracle $\mathcal{O}$ is specified, we write $M^{\mathcal{O}}$. While querying an oracle takes a single step, writing down the string to test does not take a single step. 

We are particularly interested in the following use of oracle Turing machines:

\begin{definition}
Let $f, g \in \textsf{\#P}$. We say that $f$ is \emph{polynomial-time Turing reducible} to $g$, denoted $f \leq_{T}^{\textsf{P}} g$, if there exists an oracle Turing machine $M^{g}$ that computes $f$ and halts in polynomial time.
\end{definition}

\begin{remark}
It is well-known that if $f \leq_{\text{wpar}}^{\textsf{P}} g$, then $f \leq_{T}^{\textsf{P}} g$ \cite{Zank1991PCompletenessVM}. When considering a notion of $\textsf{\#P}$-completeness, we want a reduction with the property that if $g$ is $\textsf{\#P}$-complete with respect to our reduction and $g \in \textsf{FP}$, then $\textsf{FP} = \textsf{\#P}$. To this end, it suffices to consider polynomial-time Turing reductions \cite{Valiant, AroraBarak}. However, in certain areas, it is sometimes desirable to preserve counts, which motivates the need for the more stringent (weak) parsimonious reductions.
\end{remark}

\begin{definition}
Let $\mathcal{C}$ be a complexity class, and let $\leq_{r}$ be the ordering induced by $r$-reductions. We say that $f$ is $\mathcal{C}$-\emph{hard} under $r$-reductions if for all $g \in \mathcal{C}$, $g \leq_{r} f$. We say that $f \in \mathcal{C}$ is \emph{complete} under $r$-reductions if (i) $f \in \mathcal{C}$, and (ii) $f$ is $\mathcal{C}$-hard under $r$-reductions. 
\end{definition}

While our $\textsf{\#P}$-completeness results will be with respect to  polynomial-time Turing reductions, many of the intermediary constructions yield either weak parsimonious or parsimonious reductions.

We now recall some basic notions of circuit complexity, which are necessary for Theorem~\ref{thm:MainMedian}. Let $k \in \mathbb{N}$. We say that a decision problem $K$ belongs to (uniform) $\textsf{AC}^{k}$ if there exist a (uniform) family of circuits $(C_{n})$ of depth $O(\log^{k} n)$ and polynomial size over the $\textsf{AND}, \textsf{OR}, \textsf{NOT}$ gates, where $\textsf{AND}$ and $\textsf{OR}$ have unbounded fan-in, such that $x \in K \iff C_{|x|}(x) = 1$. The complexity class $\textsf{TC}^{k}$ is defined identically as $\textsf{AC}^{k}$, except that the circuits now also have access to $\textsf{Majority}$ gates of unbounded fan-in. Observe that $\textsf{AC}^{k} \subseteq \textsf{TC}^{k}$. It is well-known that
\[
\textsf{AC}^{0} \subsetneq \textsf{TC}^{0} \subseteq \textsf{L} \subseteq \textsf{NL} \subseteq \textsf{AC}^{1}.
\]  

\noindent The $\textsf{NC}$ hierarchy is
\[
\textsf{NC} = \bigcup_{k \in \mathbb{N}} \textsf{AC}^{k} = \bigcup_{k \in \mathbb{N}} \textsf{TC}^{k}.
\]

\noindent It is known that $\textsf{NC} \subseteq \textsf{P}$, and it is believed that this containment is strict.

\section{Enumerating All-Crowns Sorting Scenarios}

Let $\mathcal{C}$ be a collection of crowns. Recall that $\Pi(\mathcal{C})$ is the set of all partitions of $\mathcal{C}$. For any partition $\pi = (\pi_1, \ldots, \pi_k)$ in $\Pi(\mathcal{C})$, let $g_i$ be the sum of the sizes of the crowns in $\pi_i$ and $p_i$ be the number of crowns in $\pi_i$.

\begin{theorem}\label{thm:allcrowns}
Consider an adjacency graph consisting entirely of a collection $\mathcal{C}=\{C_1, C_2, \ldots, C_q\}$ of $q$ crowns where $C_i$ has size $c_i$. 
Then  
\begin{equation}\label{eqn:allcrowns}
\#\text{MPS}(\mathcal{C}) = \sum_{\pi \in \Pi(\mathcal{C})} \binom{d}{g_1+p_1,  \ldots, g_{k} +p_{k}}  2^{d+q-3k} \prod_{i=1}^q c_i \left(\prod_{j=1}^{k} \left(\prod_{\ell=0}^{p_j-2} (g_j+\ell) \right)\right)
\end{equation}
where $d := d(\mathcal{C})$ and, for each $\pi\in \Pi(\mathcal{C})$, $k=k(\pi)$ is the number of parts in $\pi$. 
\end{theorem}

\begin{proof}
See Appendix~\ref{app:AllCrowns}.
\end{proof}

We will now discuss some consequences that will be useful later.

\begin{corollary}\label{lem:crowns1}
Let $\mathcal{C}=\{C_1,C_2,\ldots,C_{2n}\}$ be a collection of $2n$ crowns for some positive integer $n$. Denote by $c_i$ the size of crown $C_i$. Suppose that $\sum_{i=1}^{2n}c_i=2p-2n$ for some odd prime $p$. The number of most parsimonious scenarios where all of the crowns sort together is
\[
2^{2p+2n-3}\prod_{i=1}^{2n}c_i\prod_{\ell=0}^{2n-2}(2p-2n+\ell).
\]
\end{corollary}

\begin{proof}
This is equation~\eqref{eqn:allcrowns} with the powers of $4$ and $2$ combined, setting $q=2n$ and $d=2p$. 
\end{proof}

\begin{corollary}\label{lem:crowns2}
Let $\mathcal{C}=\{C_1,C_2,\ldots,C_{2n} \}$ be a collection of $2n$ crowns for some positive integer $n$. Denote by $c_i$ the size of crown $C_i$. Suppose that $\sum_{i=1}^{2n}c_i=2p-2n$ for some odd prime $p$. Suppose $\mathcal{C}$ is partitioned by $(\pi_1, \pi_2)$ with $|\pi_{1}| = |\pi_{2}|$ such that the crowns in each $\pi_i$ sort together and the sum of the sizes of the crowns in $\pi_{i}$ is $p-n$. Then
\[
\#\text{MPS}(\mathcal{C},(\pi_1,\pi_2)) = \binom{2p}{p}2^{2p+2n-6}\prod_{i=1}^{2n}c_i\prod_{\ell=0}^{n-2}(p-n+\ell)^2.
\]
\end{corollary}
\begin{proof}
The desired count is one term in equation~\eqref{eqn:allcrowns} arising from a partition with $k=2$ parts. Since both parts have sum $p-n$, we have $g_1=g_2=p-n$. Hence, the multinomial coefficient becomes $\binom{2p}{p}$. The result follows from combining the powers of $4$ and $2$ along with the recognition that $k=2$, $q=2n$ and $d=2p$.
\end{proof}

\begin{lemma}\label{lem:partdivis}
Let $A=\{a_1, a_2, \ldots, a_n\}$ be a multiset whose elements are each at least $3$ and sum to $kp$, with $p$ an odd prime and $k$ a positive integer. Let $\mathcal{C}=\{C_1, C_2, \ldots, C_n\}$ be a collection of $n$ crowns where $C_i$ has size $a_i-1$. Let $\pi$ be a partition of $\mathcal{C}$. If $\#$MPS($\mathcal{C},\pi$) is not divisible by $p$ then, in the corresponding partition $\pi'$ of $A$, (i) each part $\pi_{i}'=\{a_{i_1}, \ldots, a_{i_q}\}$ has $q \leq p$ elements, and (ii) $p \mid \sum_{j=1}^q a_{i_j}$.
\end{lemma}

\begin{proof}
Let $m := |\pi|$. Consider an arbitrary part of $\pi'$, labeled without loss of generality as $\pi_{1}'=\{a_{1}, \ldots, a_{q}\}$. Set $c := \sum_{i=1}^{q} a_{i}$. We aim to show that if $p$ does not divide \begin{equation}\label{eqn:MPSpi}
\#\textrm{MPS}(\mathcal{C},\pi) =
\binom{kp}{g_1+p_1, \ldots, g_m +p_m}  2^{kp+n-3m} \prod_{i=1}^n c_i \left(\prod_{j=1}^m \left(\prod_{l=0}^{p_j-2} g_j+l \right)\right),
\end{equation} 

\noindent then (i) $q \leq p$ and (ii) $p|c$ , where \eqref{eqn:MPSpi} follows from \eqref{eqn:allcrowns}.

Indeed, since $p_1=q$, the product $\prod_{l=0}^{q-2}(g_1+l)$ is a factor of $\#$MPS($\mathcal{C},\pi$). This product of $q-1$ consecutive integers is divisible by each of $1,\ldots, q-1$. If $p$ does not divide~\eqref{eqn:MPSpi}, and hence not $\prod_{l=0}^{q-2}(g_1+l)$, we must have that $q-1<p$. This establishes that $q \leq p$.

Now, observe that
\[
\binom{kp}{g_1+p_1, g_2+p_2, \ldots, g_m +p_m} = 
\binom{kp}{g_1+p_1}\binom{kp-g_1-p_1}{g_2+p_2, \ldots, g_m +p_m}.
\]

\noindent As $g_{1} + p_{1} = c$, we have that if $p$ does not divide (\ref{eqn:MPSpi}), then $p \nmid \binom{kp}{c}$. We claim this implies $p\mid c$. To see this, we show the contrapositive, and assume $p\nmid c$. Writing $c$ in base $p$ yields $c= c_0+c_1p+\cdots+c_np^n$, where $c_0\neq 0$ since $p\nmid c$. Note that the units digit of $kp$ in base $p$ is $0$, since $p|kp$. Hence subtracting $c$ from $kp$ in base $p$ requires at least one borrow, and we have that $\binom{kp}{c}\equiv 0 \pmod p$ by Kummer's Theorem \cite{Kummer}, yielding $p\mid \binom{kp}{c}$ as desired.
\end{proof}

\section{\algprobm{Pairwise Rearrangement} is $\textsf{\#P}$-complete}

In this section, we establish the following result, restated from Section~\ref{sec:Introduction}:

\begin{thm:MainCounting} 
In the Single Cut-and-Join model, the
\algprobm{Pairwise Rearrangement} problem is \mbox{$\#\textsf{P}$-complete} under polynomial-time Turing reductions.
\end{thm:MainCounting}

\begin{remark}
We establish \Thm{thm:SCJCount} in the case when the adjacency graph consists only of crowns. 
\end{remark}

Our starting point in establishing Theorem~\ref{thm:SCJCount} is the \algprobm{Multiset-Partition} problem, which is known to be $\textsf{\#P}$-complete under parsimonious reductions \cite{Simon}. Precisely, the \algprobm{Multiset-Partition} problem takes as input a multiset $A = \{ a_{1}, \ldots, a_{n}\}$ with each $a_{i} > 0$ is an integer. The goal is to count the number of ways of partitioning $A$ into two multisets $B = \{ b_{1}, \ldots, b_{j}\}$ and $ C = \{ c_{1}, \ldots, c_{k}\}$ with equal sum, i.e. such that
\[
\sum_{i=1}^{j} b_{i} = \sum_{i=1}^{k} c_{i}.
\]

\noindent We say that $B$ and $C$ have equal size if $j = k$. The \algprobm{Multiset-Equal-Partition} is a variant of \algprobm{Multiset-Partition} in which we wish to count partitions of $A$ into two multisets in which both the sizes and sums are equal. We first establish the following:

\begin{proposition}
\algprobm{Multiset-Equal-Partition} is $\textsf{\#P}$-complete under polynomial-time Turing reductions.
\end{proposition}

\begin{proof}
Clearly, \algprobm{Multiset-Equal-Partition} is in \textsf{\#P}. To show that it is \textsf{\#P}-hard, we provide a polynomial-time Turing reduction from \algprobm{Multiset-Partition} to \algprobm{Multiset-Equal-Partition}.

Let $A = \{ a_{1}, \ldots, a_{n}\}$ be an instance of \algprobm{Multiset-Partition}. Let $t := \sum_{i \in [n]} a_{i}$. For each $0 \leq k \leq n-2$, define the multiset $A_{k}$ as follows:
\[
A_{k} = \{a_1, \ldots, a_{n}, \underbrace{t,t,\ldots, t}_{k+1 \text{ copies}}, (k+1)t\}.
\]

Let $n_{k}$ denote the number of equal-size, equal-sum partitions of $A_{k}$. For $0 \leq k \leq n-2$, let $m_{k}$ be the number of equal-sum partitions $\{B,C\}$  of $A$ where $k$ is the absolute value of $|C|-|B|$. We will show that $m_{0} = n_{0}/2$, and for $1 \leq k \leq n-2$ we will show that $m_{k} = n_{k}$. As a result, \[N=\frac{n_{0}}{2} + \sum_{k=1}^{n-2} n_{k}\]  is the number of equal-sum partitions of $A$ and a quantity that we can compute in polynomial time relative to the \algprobm{Multiset-Equal-Partition} oracle.
 
When $k = 0$, note that $n_{0}$ is the number of equal-size, equal-sum partitions $\{ B, C\}$ of $A_{0}~=~\{ a_{1}, \ldots, a_{n}, t, t\}$. By our choice of $t$, it is necessary that $B$ and $C$ each contain a copy of $t$. As the elements of $A_{0}$ are distinguishable, we obtain a second partition by swapping the copies of $t$ between $B$ and $C$. Thus, each equal-size, equal-sum partition of $A$ corresponds to two equal-size, equal-sum partitions of $A_{0}$. It follows that $m_0 = n_{0}/2$, as desired.

For $k > 0$, we will show that the equal-sum partitions $\{B, C\}$ of $A$, with $k$ being the absolute value of $|C| - |B|$, are in bijection with the equal-size, equal-sum partitions of $A_{k}$. Fix an equal-sum partition $\{B, C\}$ of $A$ with $k$ as defined in the previous sentence. By adding $(k+1)$ copies of $t$ to $B$ and a single copy of $(k+1)t$ to $C$, we obtain an equal-size, equal-sum partition of $A_k$. So $m_{k} \leq n_{k}$. We now argue that all equal-size, equal-sum partitions $\{ B_{k}, C_{k}\}$ are of the form
\begin{equation*}
B_{k} = B \cup \{ \underbrace{t,t,\ldots, t}_{k+1 \text{ copies}} \} \text{ and }
C_{k} = C \cup \{ (k+1)t\},
\end{equation*}

\noindent where $\{ B, C\}$ is an equal-sum partition of $A$ with $k = |C| - |B|$. For contradiction, suppose $\{B'_{k}, C'_{k}\}$ is an equal-size, equal-sum partition of $A_{k}$ that is not of this form. Without loss of generality, suppose $(k+1)t$ belongs to $C'_{k}$. Since $\{B'_{k}, C'_{k}\}$ does not have the desired form, it is necessary that at least one copy of $t$ belongs to $C'_{k}$. Therefore, the sum of the elements in $C'_k$ is at least $(k+2)t$ while the sum of the elements in $B'_k$ is at most $kt+(t-1)$. Therefore $\{B'_k,C'_k\}$ is not an equal-sum partition, a contradiction. So $m_k=n_k$.
\end{proof}

\noindent In order to establish Theorem~\ref{thm:SCJCount}, we will establish a polynomial-time Turing reduction from \algprobm{Multiset-Equal-Partition} to \algprobm{Pairwise Rearrangement} in the SCaJ model. Let $A$ be an instance of \algprobm{Multiset-Equal-Partition}. Our first step will involve a series of transformations from $A$ to a new multiset $A''$, where (i) every equal-sum partition of $A''$ has equal size, and (ii) the number of equal-size, equal-sum partitions of $A''$ is twice that of $A$. Then starting from $A''$, we will construct a polynomial-time Turing reduction to \algprobm{Pairwise Rearrangement} in the SCaJ model. 

Without loss of generality, we may assume that $|A| = 2n$ for some integer $n$, and that $\sum_{x \in A} x$ is even. Otherwise, $A$ has no equal-size, equal-sum partitions. Let $a := 1 + \sum_{x \in A} x$. Define $A'$ to be the multiset obtained by adding $a$ to every element of $A$. We will later construct a multiset $A''$ from $A'$. We begin by establishing the following properties about $A'$:

\begin{claim}
$A'$ has the same number of equal-size, equal-sum partitions as $A$.
\end{claim}

\begin{proof}
Observe that any equal-size, equal-sum partition of $A$ corresponds to a unique one in $A'$ where each element is increased by $a$. Similarly, every equal-size, equal-sum partition of $A'$ corresponds to a unique one in $A$, where each element is decreased by $a$. 
\end{proof}

\begin{claim} \label{Claim2}
Every partition of $A'$ with equal sum must also have equal size.
\end{claim}

\begin{proof}
 The proof is by contrapositive. Let $\{B',C'\}$ be a partition of $A'$ with $|B'|<|C'|$. Then for some integer $t \geq 1$, we have that $|B'|=n-t$ and $|C'|=n+t$. By construction, every element of $A'$ is greater than $a$. Thus, the sum of the elements of $B'$ is greater than $(n-t)a$ and less than $(n-t+1)a$. On the other hand, the sum of the elements of $C'$ is greater than $(n+t)a$, which is greater than $(n-t+1)a$. Therefore, this partition of $A'$ does not have parts with equal sum, as required.   
\end{proof}

Given $A'$, we now construct a second multiset $A''$. Let $b := \sum_{x \in A'} x = (2n+1)a-1$. As $a$ is odd, we have that $b$ is even. Choose an integer $c > b$ such that $b + 2c = 2p$, for some prime $p > \binom{2n+2}{n+1}$. (We will later construct a set $\mathcal{C}$ of crowns such that the pair $(p, \mathcal{C})$ satisfy the hypotheses of Corollary~\ref{lem:crowns1}.) Observe that $p \nmid c$. Define $A'' := A' \cup \{ c, c\}$. Note that $\sum_{x \in A''} x = 2p$.

\begin{claim}
Every partition of $A''$ with equal sum must also have equal size. In particular, the number of equal-size, equal-sum partitions of $A''$ is twice the number of equal-size, equal-sum partitions of $A'$. 
\end{claim}

\begin{proof}
 Since the sum of all non-$c$ elements of $A''$ is less than $c$, the two copies of $c$ must be in opposite parts of any equal-sum partition of $A''$. The non-$c$ elements of such a partition form an equal-sum partition of $A'$, which must be equal-size by Claim~\ref{Claim2}. Thus, for every equal-size, equal-sum partition $\{B'', C''\}$ of $A''$, there exists an equal-size, equal-sum partition $\{B', C'\}$ of $A'$ such that $B'' = B' \cup \{c\}, C'' = C' \cup \{c\}$. Furthermore, as the elements of a multiset are distinguishable, every equal-size, equal-sum partition of $A'$ yields in fact two equal-size, equal-sum partitions of $A''$ (by exchanging the copies of $c$ between the two parts). The result now follows.
\end{proof}

We now construct a set of crowns $\mathcal{C}$ from $A''$ as follows. Let $\mathcal{C} = \{ C_{1}, \ldots, C_{2n+2}\}$ be a set of $2n+2$ crowns, one with size $m-1$ for each $m \in A''$. Observe that the (multipartite) partitions of $A''$ are in bijection with the (multipartite) partitions of $\mathcal{C}$ (where we stress that each crown in its entirety belongs to a single part; that is, we are partitioning the set $\mathcal{C}$ of crowns and not the underlying genes or telomeres). 

Denote $c_{i}$ to be the size of $C_{i}$. Let $N$ be the number of sorting scenarios for $\mathcal{C}$. Let
\[
N' = N - 2^{2p+2n-1}\prod_{i=1}^{2n+2}c_i\prod_{\ell=0}^{2n}(2p-2n-2+\ell).
\]

\noindent By Corollary~\ref{lem:crowns1}, $N'$ equals the number of most parsimonious scenarios for the $2n+2$ crowns in $\mathcal{C}$ where the crowns do not all sort together. Let $w\equiv N'\pmod p$ with $0\leq w <p$. We will show later that $p \nmid N'$ (see Claim~\ref{Claim3.11}). Now let
\[
M = \binom{2p}{p}2^{2p+2n-4}\prod_{i=1}^{2n+2}c_i\prod_{\ell=0}^{n-1}(p-n-1+\ell)^2.
\]

\noindent By Corollary~\ref{lem:crowns2}, $M$ equals the number of most parsimonious scenarios for any partition of the crowns into two parts of equal size, where the sum of the sizes of the crowns in each part is the same and the crowns within the same part sort together. Let $u\equiv M\pmod p$ with $0\leq u <p$.  We will now show that $u > 0$.

\begin{claim}
$M$ is not divisible by $p$. Consequently, $u > 0$.    
\end{claim}

\begin{proof}
By the Lucas congruence for binomial coefficients, we have that $p \nmid \binom{2p}{p}$. As $p$ is odd, $p \nmid \, 2^{2p+2n-4}$. Similarly, as $p > \binom{2n+2}{n+1}$, we have that for any $0 \leq \ell \leq n-1$ that $p \nmid (p-n-1+\ell)$. It remains to show that for each $i \in [2n+2]$, $p \nmid c_{i}$. As $b + 2c = 2p$ and $c < p$, we have necessarily that $p \nmid c$. Recall that there were two crowns of size $c-1$, say $C_{2n+1}$ and $C_{2n+2}$. So we have that $p \nmid c_{2n+1}$ and $p \nmid c_{2n+2}$. By the choice of $c$, namely that $c_{i} \leq b < p$, we have that $p \nmid c_{i}$. As $p$ is prime and does not divide any factor of $M$, it follows that $p$ does not divide $M$.
\end{proof}

In order to complete the reduction and the proof of Theorem~\ref{thm:SCJCount}, we will recover the number of equal-size, equal-sum partitions of $A''$ in polynomial-time, relative to the \algprobm{Pairwise Rearrangement} oracle.

\begin{claim} \label{Claim3.11}
The number of equal-size, equal-sum partitions of $A''$ is the least positive integer $m$ satisfying $m\equiv u^{-1}w \pmod{p}$.
\end{claim}

\begin{proof} 
As the total number of elements in $A''$ is less than $p$, any part of any partition of $A''$ has fewer than $p$ elements. So by Lemma~\ref{lem:partdivis}, the only terms in the formula from Theorem~\ref{thm:allcrowns} for $\mathcal{C}$ that are nonzero modulo $p$ are the terms corresponding to (multipartite) partitions of $A''$ where $p$ divides the sum of the elements in each part. The only possibilities for this are the trivial partition of $A''$ with one part, which has sum $2p$, and the partitions of $A''$ into two parts with equal-size and equal-sum. Recall that $N-N'$ is the number of sorting scenarios for $\mathcal{C}$ corresponding to the trivial partition of $A''$ and $M$ is the number of sorting scenarios of $\mathcal{C}$ corresponding to each equal-size, equal-size partition of $A''$ into two parts.

Since $N'$ is the number of scenarios when the crowns do not all sort together, when we calculate $N'$ modulo $p$, the only nonzero terms correspond to the equal-size equal-sum partitions of $A''$ into two parts and each of these has $M$ sorting scenarios. Therefore \mbox{$N'\equiv vM\pmod{p}$}, where $v$ is the number of equal-size equal-sum partitions of $A''$ into two parts. Since $w\equiv N'\pmod p$ and $u\equiv M \pmod p$, then $w=vu\pmod p$. Recall that $0<u<p$, so $v=u^{-1}w\pmod{p}$. Further, since $|A''|=2n+2$, $v$ is at most $\binom{2n+2}{n+1} < p$, the number of equal-size partitions of $A''$, the relation $v\equiv u^{-1}w\pmod{p}$ uniquely determines $v$, as required.   
\end{proof}

\section{Counting Single Cut or Join Medians}

In this section we turn our attention to the Single Cut or Join model and improve upon a result of Mikl\'os \& Smith \cite{MiklosSmithSamplingCounting} to establish a slightly stronger complexity result for enumerating medians. The following theorem is restated from Section~\ref{sec:Introduction}:
\begin{thm:MainMedian}
In the Single Cut or Join model, the $\#\algprobm{Median}$ problem belongs to $\textsf{FL}$.
\end{thm:MainMedian}

\begin{proof}
Let $\mathcal{G} := \{ G_{1}, G_{2}, \ldots, G_{k} \}$ be a set of genomes with the same set of edge-labels.  Mikl\'os \& Smith \cite[Theorem~1]{MiklosSmithSamplingCounting} previously showed this problem to be in $\textsf{FP}$. A careful analysis of their work shows this problem to be in $\textsf{FL}$. We first note that the adjacencies and edge labels uniquely determine a genome. Thus, we may assume a genome is represented as an adjacency matrix, which requires $O(n^{2} \log n)$ bits to specify.

Feijão \& Meidanis \cite{FeijaoMeidanis} showed that in the case when $k = 3$, there is a unique median and this median is obtained by selecting precisely the adjacencies that occur in a majority of the genomes. Their proof trivially extends to all odd $k$. Therefore, for the remainder of the proof we will assume $k$ is even. 

When $k$ is even, let $H$ be the set of adjacencies which occur in exactly half of the genomes in $\mathcal{G}$. However some of the adjacencies in $H$ share an extremity. So any optimal median for $\mathcal{G}$ contains the all of the adjacencies that are present in more than half of the genomes in $\mathcal{G}$ and any conflict-free subset of adjacencies in $H$. In order to enumerate the number of optimal medians, Mikl\'os \& Smith \cite[Theorem~1]{MiklosSmithSamplingCounting} introduced the \emph{conflict graph} $C$ which is defined as follows. The vertices of $C$ is the set of extremities present in each genome in $\mathcal{G}$. Then, there is an undirected edge $\{v_{1}, v_{2} \}$ precisely when the adjacency $v_{1}v_{2}$ is in $H$.  Mikl\'os \& Smith \cite[Observation~1]{MiklosSmithSamplingCounting} observed that each vertex in $C$ has degree at most $2$, and so the components of $C$ are paths and cycles. Any conflict-free subset of the adjacencies is a matching (not necessarily maximal and possibly empty) of $C$. So it suffices to count the matchings in each component of $C$, and then multiply these counts. With this result in mind, let's look more clarify at the complexity.

To record the conflict graph, the vertices can be written in $\textsf{AC}^{0}$ by copying the label for each extremity. Now for a pair of extremities $v_{1}, v_{2}$, we can check in $\textsf{TC}^{0}$ if the adjacency $v_{1}v_{2}$ is present in $H$ and then add the edge if it is.

Now given the conflict graph, we use Reingold's $\textsf{FL}$ undirected connectivity algorithm~\cite{Reingold} to write down the connected component as follows. We iterate over the vertices, starting with an arbitrary $u$. We begin by writing down $u$ to the output tape. Now for each vertex $v \neq u$, we use Reingold's connectivity algorithm \cite{Reingold} to check whether $u$ and $v$ belong to the same connected component and output $v$ if they do. At any step, we are using $O(\log n)$ bits to store $u$ and $v$ and $O(\log n)$ bits to run Reingold's algorithm \cite{Reingold}. So writing down the connected components (with multiplicity) is computable using a logspace transducer. 
    
Now given the connected components, we use a second logspace transducer to identify duplicates. For each connected component $X$, we pick a vertex $v_{X}$ in $X$. Then we iterate over the previously considered connected components. If $v_{X}$ appears on any previous connected component on the input tape, we move on to the next component. Otherwise, we write the current component to the output tape. Thus, the computations up to this stage are all $\textsf{FL}$-computable.

We now proceed to compute the number of optimal medians. For each connected component, it is well-known \cite{LovaszPlummer} that the number of matchings in a path of length $\ell$ is
    \begin{align} \label{MatchingPath}
    \sum_{k=1}^{\lfloor \ell/2 \rfloor} \binom{\ell-k}{k},
    \end{align}
    
  \noindent  and the number of matchings in a cycle of length $\ell$ is
    \begin{align} \label{MatchingCycle}
    \sum_{k=1}^{\lfloor \ell/2 \rfloor} \frac{\ell}{\ell-k} \binom{\ell-k}{k}.
    \end{align}
    
Since there are $n$ vertices and $O(n)$ edges, our input requires space $O(n^{2} \log n)$. Each of $1, \ldots, n$ can be written down using $O(\log n)$ bits. Now iterated multiplication circuit of $m$ $m$-bit integers is $\textsf{TC}^{0}$-computable \cite{HESSE2002695}. In our setting, we have that $m = O(n^{2} \log n)$. As only $O(n \log n)$ bits are required to write down the sequence $1, \ldots, n$ and as $k \leq n$, we can compute $n!$ and $k!$ in $\textsf{TC}^{0}$. Furthermore, division is $\textsf{TC}^{0}$-computable \cite{HESSE2002695}. Thus, we may compute $\binom{n}{k}$ and $\binom{n-k}{k}$ in $\textsf{TC}^{0}$. 
So each term in the summations from (\ref{MatchingPath}) and (\ref{MatchingCycle}) requires $O(n \log n) \subseteq O(n^{2} \log n)$ bits to write down. It is well-known that the iterated addition of $m$ $m$-bit integers is $\textsf{TC}^{0}$-computable (here again, $m = O(n^{2} \log n)$) (c.f., \cite{VollmerText}), thus we may evaluate the summations in $\textsf{TC}^{0}$. 
    
Finally, we multiply the number of matchings on each connected component. We note that each component has at most $2^{n}$ matchings, and so writing down the number of matchings for each component takes $O(n)$ (which is fewer than $O(n^{2} \log n)$) bits. Again, as iterated multiplication is $\textsf{TC}^{0}$-computable \cite{HESSE2002695}, we may use an $\textsf{TC}^{0}$ circuit to compute the final output, which is the number of medians.
    
    We have a constant number of $\textsf{FL}$-computable stages, and so our algorithm is \mbox{$\textsf{FL}$-computable} in this case.
The result now follows.
\end{proof}

In light of the discussion in \Rmk{rmk:MedianComplexity}, it is natural to ask if the upper bound of $\textsf{FL}$ for \Thm{thm:MedianSCJ} can be improved. The proof relies crucially on identifying connected components in a graph, which is $\textsf{FL}$-complete, even when the graph is a forest with all vertices having degree at most $2$ \cite{CookMcKenzie}. Thus, we conjecture that our bound is optimal---namely, that $\algprobm{\#Median}$ is in fact $\textsf{FL}$-complete.

\section{Conclusion}

We investigated the computational complexity of genome rearrangement problems in the Single Cut-and-Join and Single Cut or Join models. In particular, we showed that the \algprobm{Pairwise Rearrangement} problem in the Single Cut-and-Join model is $\#\textsf{P}$-complete under polynomial-time Turing reductions (Theorem \ref{thm:MainCounting}) and in the Single Cut or Join model, $\#\algprobm{Median} \in \textsf{FL}$ (Theorem~\ref{thm:MedianSCJ}). 

Natural next steps for the Single Cut-and-Join model include investigating the complexity of  $\algprobm{Median}$-- it remains open whether this problem is $\textsf{NP}$-hard. It would also be interesting to investigate whether there is particular combinatorial structure in the Single Cut-and-Join model that forces the \algprobm{Pairwise Rearrangement} problem to be $\#\textsf{P}$-complete. It remains open whether \algprobm{Pairwise Rearrangement} remains $\textsf{\#P}$-hard when the adjacency graph does not contain any crowns. In light of our result that \algprobm{Pairwise Rearrangement} is $\#\textsf{P}$-complete, we seek efficient algorithmic approaches to cope with this intractability. One natural approach would be to investigate if we can efficiently sample sorting scenarios; that is, whether \algprobm{Pairwise Rearrangement} belongs to $\textsf{FPRAS}$.  Along with these natural next steps, we point the reader to Table~\ref{fig:Fig1} for additional open questions regarding the complexity of these combinatorial problems. 

\bibliographystyle{plainurl}
\bibliography{refs}

\newpage
\appendix
\section{Proof of Theorem~\ref{thm:allcrowns}} \label{app:AllCrowns}

\noindent In this section, we will prove Theorem~\ref{thm:allcrowns}. We will begin with some preliminary lemmas.

\begin{lemma}\label{lem:ksum}
For any integer $i\geq0$, integer $q\geq3$, integer $1\leq t\leq i$, and real number $s$,
\[
\sum_{k=i-t}^{i-1}\pb{\pb{s+\pb{i-k}q+\pb{k-2i}}\prod_{m=1}^{q-3}\pb{s+m-k}}=t\prod_{m=1}^{q-2}\pb{s+m-i+t}.
\]
\end{lemma}

\begin{proof}
We proceed by induction on $t$. \noindent When $t=1$,
\begin{align*}
\sum_{k=i-1}^{i-1}&\pb{\pb{s+\pb{i-k}q+\pb{k-2i}}\prod_{m=1}^{q-3}\pb{s+m-k}}\\
&=\pb{s+\pb{i-\pb{i-1}}q+\pb{\pb{i-1}-2i}}\prod_{m=1}^{q-3}\pb{s+m-i+1}
=\prod_{m=1}^{q-2}\pb{s+m-i+1}, 
\end{align*}

\noindent as required. Now, suppose that
\[
\sum_{k=i-t+1}^{i-1}\pb{\pb{s+\pb{i-k}q+\pb{k-2i}}\prod_{m=1}^{q-3}\pb{s+m-k}}=\pb{t-1}\prod_{m=1}^{q-2}\pb{s+m-i+t-1}.
\]
We have,
\begin{align*}
\sum_{k=i-t}^{i-1}&\pb{\pb{s+\pb{i-k}q+\pb{k-2i}}\prod_{m=1}^{q-3}\pb{s+m-k}}\\
&=\pb{s+\pb{i-\pb{i-t}}q+\pb{\pb{i-t}-2i}}\prod_{m=1}^{q-3}\pb{s+m-i+t}\\
&\quad +\sum_{k=i-t+1}^{i-1}\pb{\pb{s+\pb{i-k}q+\pb{k-2i}}\prod_{m=1}^{q-3}\pb{s+m-k}}\\
&=\pb{s+\pb{i-\pb{i-t}}q+\pb{\pb{i-t}-2i}}\prod_{m=1}^{q-3}\pb{s+m-i+t}+\pb{t-1}\prod_{m=1}^{q-2}\pb{s+m-i+t-1}\\
&=\pb{s+tq-i-t}\prod_{m=1}^{q-3}\pb{s+m-i+t}+\pb{t-1}\prod_{m=0}^{q-3}\pb{s+m-i+t}
=t\prod_{m=1}^{q-2}\pb{s+m-i+t},
\end{align*}

\noindent as required.
\end{proof}

\begin{lemma}\label{lem:crownidentity}
For integer $q\geq3$ and nonnegative integers $c_1,c_2,\ldots,c_q$, define $s=\sum_{i=1}^qc_i$. Then,
\[
\sum_{j=1}^q\sum_{k=0}^{c_j-1}\pb{\pb{\pb{q-2}c_j-\pb{q-1}k+s}\prod_{m=1}^{q-3}\pb{s-k+m}}=\prod_{m=2}^q\pb{s+q-m}.
\]
\end{lemma}

\begin{proof}
Note that any two values $j_1$ and $j_2$ with $c_{j_1}=c_{j_2}$ contribute the same value to the sum over $j$. With this in mind, define a sequence $\st{a_i}_{i\geq0}$, where $a_i$ equals the number of values among $c_1,c_2,\ldots,c_q$ equal to $i$. We can now regroup the given sum by these frequencies:
\begin{align*}
\sum_{j=1}^q&\sum_{k=0}^{c_j-1}\pb{\pb{\pb{q-2}c_j-\pb{q-1}k+s}\prod_{m=1}^{q-3}\pb{s-k+m}}\\
&=\sum_{i=0}^\infty a_i\sum_{k=0}^{i-1}\pb{\pb{\pb{q-2}i-\pb{q-1}k+s}\prod_{m=1}^{q-3}\pb{s-k+m}}\\
&=\sum_{i=0}^\infty a_i\sum_{k=0}^{i-1}\pb{\pb{s+\pb{i-k}q+\pb{k-2i}}\prod_{m=1}^{q-3}\pb{s-k+m}}.
\end{align*}
By Lemma~\ref{lem:ksum} (taking $t=i$), we have
\begin{align*}
\sum_{i=0}^\infty &a_i\sum_{k=0}^{i-1}\pb{\pb{s+\pb{i-k}q+\pb{k-2i}}\prod_{m=1}^{q-3}\pb{s-k+m}}\\
&=\sum_{i=0}^\infty a_i\cdot i\prod_{m=1}^{q-2}\pb{s+m} 
=\prod_{m=1}^{q-2}\pb{s+m}\sum_{i=0}^\infty a_i\cdot i.
\end{align*}
Since $a_i$ is the frequency of the value $i$ among $c_1,c_2,\ldots,c_q$, we observe that
\[
\sum_{i=0}^\infty a_i\cdot i=\sum_{i=1}^qc_i=s.
\]
So, we have
\[
\prod_{m=1}^{q-2}\pb{s+m}\sum_{i=0}^\infty a_i\cdot i=s\prod_{m=1}^{q-2}\pb{s+m}=\prod_{m=0}^{q-2}\pb{s+m}.
\]
Replacing the index $m$ by $q-m$, this equals the required product of $\prod\limits_{m=2}^q\pb{s+q-m}$.
\end{proof}

\begin{proposition}  \label{prop:allcrownjoin}
Consider an adjacency graph consisting entirely of a collection $\mathcal{C}=\{C_1, C_2, \ldots, C_q\}$ of crowns where $C_i$ has size $c_i$. If all crowns must sort together, then the number of most parsimonious scenarios is
$$\left(\prod\limits_{i=1}^q c_i\right)2^{d+q-3}\prod\limits_{k=2}^q (d-k).$$
\end{proposition}

\begin{proof}
Consider a sorting scenario in which all crowns sort together. 
Initially, a cut operation replaces one of the crowns with a $W$-shaped component. Because all the crowns sort together, the case analysis in Observation~\ref{obs:SortingScenarios} reveals that, for each other crown, there is a unique cut-join operation that cuts an adjacency in the crown and joins it with a telomere from a $W$-shaped component. Thus each scenario induces a linear order on the crowns based on the first cut or cut-join operation that operated on  an adjacency in each crown.  

    First consider the case when the crowns get cut in order $C_1, C_2, \ldots, C_q$. After cutting $C_1$, we are left with a $W$-shaped component of size $c_1$. Applying $k_1$ ($0\leq k_1\leq c_1-1$)  sorting steps to this $W$-shaped component results in a $W$-shaped component of size $c_1-k_1$. 
    Now we cut-join $C_2$ into this $W$-shaped component of size $c_1-k_1$ which yields $W$-shaped component of size $c_1+c_2-k_1$. Applying $k_2$ sorting steps to this new component yields a $W$-shaped component of size $c_1+c_2-k_1-k_2$ $(0\le k_2\le c_1+c_2-k_1-1)$. Then cut-join $C_3$ into this $W$-shaped component, continuing in this manner until everything is sorted.

    This gives us the following number of most parsimonious scenarios:
    \begin{align*}
        &\left(\prod\limits_{i=1}^q c_i\right)2^{\left(\sum\limits_{i=1}^q c_i\right)-1+2(q-1)} 
        \sum\limits_{k_1=0}^{(c_1-1)}\;\sum\limits_{k_2=0}^{(c_1+c_2-k_1-1)}\ldots
        \sum\limits_{k_{q-1}=0}^{\left[\left(\sum\limits_{i=1}^{q-1} c_i\right)-\left(\sum\limits_{j=1}^{q-2}k_j\right)-1\right]}
        \sum\limits_{k_q=0}^{\left[\left(\sum\limits_{i=1}^q c_i\right)-\left(\sum\limits_{j=1}^{q-1}k_j\right)\right]} 1\\
        =\; &\left(\prod\limits_{i=1}^q c_i\right)2^{(d-q)+2q-3} 
        \sum\limits_{k_1=0}^{(c_1-1)}\sum\limits_{k_2=0}^{(c_1+c_2-k_1-1)}\ldots
        \sum\limits_{k_{q-1}=0}^{\left[\left(\sum\limits_{i=1}^{q-1} c_i\right)-\left(\sum\limits_{j=1}^{q-2}k_j\right)-1\right]}
        \left[\left(\sum\limits_{i=1}^q c_i\right)-\left(\sum\limits_{j=1}^{q-1}k_j\right)\right]\\
        =\; &\left(\prod\limits_{i=1}^q c_i\right)2^{d+q-3}
        \sum\limits_{k_1=0}^{(c_1-1)}\sum\limits_{k_2=0}^{(c_1+c_2-k_1-1)}\ldots\\
        &  
        \sum\limits_{k_{q-1}=0}^{\left[\left(\sum\limits_{i=1}^{q-2} c_i\right)-\left(\sum\limits_{j=1}^{q-3}k_j\right)-1\right]}
        \left[\frac{1}{2}\left(\sum\limits_{i=1}^{q-2} c_i - \sum\limits_{j=1}^{q-3} k_j\right)\left(c_{q-1}+1+\sum\limits_{i=1}^{q-1} c_i-\sum\limits_{j=1}^{q-2}k_j\right)\right]\\
        =\; & \qquad \vdots\\
        =\; & \left(\prod\limits_{i=1}^q c_i\right)2^{d+q-3}
        \sum\limits_{k_1=0}^{(c_1-1)}
        \left\lbrace
        \left((q-2)c_1-(q-1)k_1+\sum\limits_{i=1}^{q} c_i \right)
        \prod\limits_{m=1}^{q-3}\left[\left(\sum\limits_{i=1}^{q} c_i\right)-k_1+m\right]
        \right\rbrace\\
    \end{align*} 
   
    \noindent Repeating this process for all cases considering cutting can start with any of the crowns, we have that the number of most parsimonious scenarios is
    $$\left(\prod\limits_{i=1}^q c_i\right)2^{d+q-3}
        \sum\limits_{j=1}^q
        \sum\limits_{k=0}^{(c_j-1)}
        \left\lbrace
        \left((q-2)c_j-(q-1)k+\sum\limits_{i=1}^{q} c_i \right)
        \prod\limits_{m=1}^{q-3}\left[\left(\sum\limits_{i=1}^{q} c_i\right)-k+m\right]
        \right\rbrace, $$
    which by Lemma \ref{lem:crownidentity} is equal to 
        $$\left(\prod\limits_{i=1}^q c_i\right)2^{d+q-3}\prod\limits_{k=2}^q (d-k).$$
\end{proof}

We now work toward a formula for counting all most parsimonious scenarios with all crowns. The formula in Theorem \ref{thm:allcrowns} above comes from Proposition \ref{prop:allcrownjoin} above and Lemma \ref{lem:partcount} below.

\begin{lemma}\label{lem:partcount}
Let $\mathcal{B}$ be the set of components of an adjacency graph. We have
\[
\#\text{MPS}(\mathcal{B})=\sum_{\pi\in\Pi(\mathcal{B})}\binom{d(\mathcal{B})}{d(\pi_1),d(\pi_2),\ldots,d(\pi_{k})}\prod_{i=1}^{k}\#\text{ST}(\pi_i),
\]

\noindent where, for each $\pi\in \Pi(\mathcal{B})$, $k=k(\pi)$ is the number of parts in $\pi$.
\end{lemma}

\begin{proof}
Since ``sort together'' is an equivalence relation (see Definition~\ref{def:equivrelation}), any most parsimonious scenario on $\mathcal{B}$ partitions the elements of $\mathcal{B}$ into sets of components that sort together. Given such a partition $\pi$, we must choose a sorting scenario for each part. The number of choices for part $\pi_i$ is $\#\text{ST}(\pi_{i})$, and the choices are independent across parts. 
 So, there are
\[
\prod_{i=1}^{k} \#\text{ST}(\pi_i)
\]
total ways to select most parsimonious scenarios for all the parts of $\pi$. Finally, we must select the order in which to do the operations in a combined sorting scenario. The total number of operations is $d(\mathcal{B})$. Of these operations, $d(\pi_i)$ must take place in part $i$ for each $i$ from $1$ to $k$. The total number of strings of length $d(\mathcal{B})$ on an alphabet with $k$ symbols and $d(\pi_i)$ copies of the $i^{th}$ symbol is equal to $\binom{d(\mathcal{B})}{d(\pi_1),d(\pi_2),\ldots,d(\pi_k)}$, and each such string corresponds to a distinct scenario. Since the choice of ordering is also independent of the scenario for each part, we multiply by the multinomial coefficient. Summing over all possible partitions yields the desired formula.
\end{proof}

Finally, we prove Theorem \ref{thm:allcrowns}.

\begin{proof}[Proof of Theorem~\ref{thm:allcrowns}]
Recall that the distance of a set of crowns equals the number of crowns plus the sizes of the crowns. By Lemma~\ref{lem:partcount},
    \[
    \#\text{MPS}(\mathcal{C})=\sum_{\pi\in\Pi(\mathcal{C})}\binom{d(\mathcal{C})}{d(\pi_1),d(\pi_2),\ldots,d(\pi_{k})}\prod_{i=1}^{k}\#\text{ST}(\pi_i).
    \]
    
    \noindent By Proposition~\ref{prop:allcrownjoin},
    \[
    \#\text{ST}(\pi_i)=\left(\prod\limits_{j=1}^{p_i} c_{{\pi_i}_j}\right)2^{d(\pi_i)+p_i-3}\prod\limits_{\ell=2}^{p_i} (d(\pi_i)-\ell),
    \]
    so
    \[
    \#\text{MPS}(\mathcal{C})=\sum_{\pi\in\Pi(\mathcal{C})}\binom{d}{d(\pi_1),d(\pi_2),\ldots,d(\pi_{k})}\prod_{j=1}^{k}\left(\prod\limits_{i=1}^{p_i} c_{{\pi_j}_i}\right)2^{d(\pi_j)+p_j-3}\prod\limits_{\ell=2}^{p_j} (d(\pi_j)-\ell).
    \]

    \noindent There are three expressions inside the product over $j$; we examine each in turn. First, we have a product over all crown lengths in $\pi_j$. When multiplied over all parts of $\pi$, we end up with a product over all crown lengths. Second, we have a power of $2$. The exponent on $2$ when everything is multiplied out will be
    \[
    \sum_{j=1}^kd(\pi_j)+\sum_{j=1}^k p_j-3k.
    \]
    This first sum equals $d$, and the second sum equals $q$, so we have
    \[
    \#\text{MPS}(\mathcal{C})=\sum_{\pi\in\Pi(\mathcal{C})}\binom{d}{d(\pi_1),d(\pi_2),\ldots,d(\pi_{k})}2^{d+q-3k}\prod_{i=1}^q c_i\prod_{j=1}^{k}\prod\limits_{\ell=2}^{p_j} (d(\pi_j)-\ell).
    \]
    Now, observe the following: 
    \begin{itemize}
        \item $d(\pi_i)=g_i+p_i$
        \item By replacing $\ell$ with $p_j-\ell$, we have
        \[
        \prod\limits_{\ell=2}^{p_j} (d(\pi_j)-\ell)=\prod\limits_{\ell=0}^{p_j-2} (d(\pi_j)-p_j+\ell) = \prod\limits_{\ell=0}^{p_j-2} (g_j+\ell).
        \]
    \end{itemize}
    So, we have
    \begin{align*}
        \#\text{MPS}(\mathcal{C}) &=\sum_{\pi\in\Pi(\mathcal{C})}\binom{d}{g_1+p_1,g_2+p_2,\ldots,g_k+p_k}2^{d+q-3k}\prod_{i=1}^q c_i\left(\prod_{j=1}^{k}\left(\prod\limits_{\ell=0}^{p_j-2} (g_j+\ell)\right)\right),
    \end{align*}
    as required.
\end{proof}

\end{document}